\documentclass{article}

\usepackage{arxiv}
\usepackage[utf8]{inputenc}
\usepackage[T1]{fontenc}
\usepackage{hyperref}
\usepackage{url}
\usepackage{booktabs}
\usepackage{amsfonts}
\usepackage{amsmath}
\usepackage{amsthm}
\usepackage{amssymb}
\usepackage{bm}
\usepackage{nicefrac}
\usepackage{microtype}
\usepackage{graphicx}
\usepackage{array}
\usepackage{multirow}
\usepackage{textcomp}
\usepackage{cite}

\graphicspath{{figures/}}

\theoremstyle{plain}
\newtheorem{proposition}{Proposition}
\newtheorem{corollary}{Corollary}

\newcommand{\Tr}{\operatorname{Tr}}
\newcommand{\R}{\mathbb{R}}
\newcommand{\E}{\mathbb{E}}
\newcommand{\PR}{\operatorname{PR}}
\newcommand{\dfull}{d_{\mathrm{full}}}
\newcommand{\ddiag}{d_{\mathrm{diag}}}
\newcommand{\dz}{d^{(Z)}}
\newcommand{\fz}{f^{(Z)}}
\newcommand{\sd}[1]{\,{\scriptsize$\pm$#1}}

\title{What Does the Effective Dimension of a Hybrid Quantum Neural Network Actually Measure?}

\author{
	Md Masum Mia \\
	Qatar University\\
	Doha, Qatar \\
	\texttt{masum.mia@qu.edu.qa} \\
	\And
	Amith Khandakar\thanks{Corresponding author: Amith Khandakar (amitk@qu.edu.qa)} \\
	Qatar University\\
	Doha, Qatar \\
	\texttt{amitk@qu.edu.qa} \\
	\And
	Mohamed Arselene Ayari \\
	Qatar University\\
	Doha, Qatar \\
	\texttt{ArslanA@qu.edu.qa} \\
}

\date{}

\renewcommand{\shorttitle}{What Does the Quantum Fisher Effective Dimension Measure?}

\hypersetup{
pdftitle={What Does the Effective Dimension of a Hybrid Quantum Neural Network Actually Measure?},
pdfauthor={Md Masum Mia, Amith Khandakar, Mohamed Arselene Ayari},
pdfkeywords={quantum Fisher information, effective dimension, hybrid quantum neural networks, barren plateaus, entanglement, quantum machine learning},
}

\begin{document}
\maketitle

\begin{abstract}
Hybrid quantum neural networks (QNNs) pass features from a classical backbone into a variational quantum circuit, and many studies log a scalar called the effective dimension as a sign of capacity or model quality. The effective dimension is usually defined as the participation ratio (PR) of the quantum Fisher information matrix (QFIM), but studies in practice log a PR of cheap parameter sensitivities instead, and the meaning of that scalar is rarely tested. We show that the logged scalar equals the PR of Pauli-$Z$ output sensitivities, that each sensitivity is at most $n$ times its QFIM diagonal entry on $n$ qubits, and that a PR of the QFIM diagonal overstates the full-matrix PR by an exact factor set by the off-diagonal mass. We then test the scalar on CIFAR-10 and EuroSAT with four frozen backbones, one of them a random projection, over 4 to 16 qubits, 1 to 4 layers, three entangling patterns, and three seeds. At a fixed circuit, CIFAR-10 accuracy spans 23.14\% to 86.44\% while the scalar stays between 6.63 and 9.13, and its correlation with accuracy changes sign across seeds in 7 of the 14 configurations. Depth raises the share of active parameters from 0.39 to 0.91, near a ceiling we derive, and initial gradient variance falls by a factor of 34.5 from 4 to 16 qubits. Removing entanglement raises accuracy by 5.3 to 8.2 points at the same parameter count. The effective dimension therefore describes circuit geometry rather than model quality, so capacity, trainability, and accuracy each need their own diagnostic.
\end{abstract}


\section{Introduction}
\label{sec:introduction}
Hybrid quantum neural networks (QNNs) attach a variational quantum circuit to a classical feature extractor and train the composite model by gradient descent \cite{schuld2021,cerezo2021,schuld2020circuit,biamonte2017}. On images, a frozen backbone compresses each picture into a short vector of rotation angles, and the circuit prepares a parameterized state from those angles. A linear head then maps the measured expectation values to class scores \cite{mari2020}. This division keeps the circuit small enough to simulate exactly, yet it creates an attribution problem. Accuracy belongs to the backbone, the bottleneck, the circuit, and the head together. A statistic computed from the circuit alone therefore need not describe the model as a whole.

The natural object for the circuit itself is the quantum Fisher information matrix (QFIM) $F$, the metric that the parameterized state induces on its parameters \cite{meyer2021,liu2020,provost1980}. Its rank counts locally independent parameter directions, and its spectrum shows how evenly the circuit's sensitivity is spread across those directions \cite{haug2021,larocca2023}. A scalar summary called the effective dimension is widely used to reason about capacity and trainability \cite{abbas2021,haug2021,larocca2023}. In practice, the same name also serves as shorthand for model quality. During training, however, the logged scalar is usually built from a cheap diagonal sensitivity estimate rather than from $F$ itself. We find a single number easier to track than a $p\times p$ matrix. That convenience helps, however, only once we know which property the number tracks.

Three recurring practices make that reading fragile in current studies of hybrid QNNs.
\begin{itemize}
\item \textit{Mixed definitions.} The name covers a global capacity measure built from the classical Fisher information \cite{abbas2021,maddox2020} and the participation ratio (PR) of the QFIM diagonal. It also covers the PR of the full matrix and diagonal estimates built from measured outputs. The relation between the logged estimate and the matrix it represents is rarely stated.
\item \textit{Single-seed trends.} Variational training is stochastic \cite{cerezo2021}, and a hybrid model adds a frozen backbone and a trained bottleneck to that randomness. A trend drawn from one seed can reverse under another seed, so the seed-to-seed spread belongs in every report.
\item \textit{Indirect plateau diagnoses.} A small gradient or a small effective dimension is sometimes taken as evidence of a barren plateau (BP). A BP is defined by how the gradient variance decays as the system grows \cite{mcclean2018,cerezo2021bp}, not by the value of a single geometric scalar.
\end{itemize}

We address these practices with a short derivation and a controlled measurement study. The derivation states exactly what the logged scalar computes and how it relates to the QFIM. We then measure it on two image datasets while the register grows from 4 to 16 qubits and the depth from 1 to 4 layers. Each comparison holds either the circuit or the features fixed, so the two effects do not mix. A random projection backbone and a product circuit without entangling gates serve as controls, and we measure the gradient variance at initialization directly instead of inferring it from geometry. As shown in Fig.~\ref{fig:overview}, the study separates the roles the scalar is asked to play and assigns each role to the quantity that can answer it.

The main contributions of this paper are the following.
\begin{itemize}
\item We define the logged effective dimension as the PR of Pauli-$Z$ output sensitivities and prove that each sensitivity is at most $n$ times the matching QFIM diagonal entry. We further prove that any diagonal PR equals the full PR divided by $1-\eta^{2}$, where $\eta$ is the off-diagonal mass of the QFIM.
\item We show that the last $R_z$ gate on each qubit never moves a Pauli-$Z$ output, which caps the number of active parameters at $p-n$. The measured fill sits within 0.04 of this ceiling whenever the circuit has two or more layers.
\item We track the scalar over 150 training epochs with three seeds, and we find that its correlation with test accuracy changes sign across seeds in 7 of the 14 configurations.
\item We hold the circuit fixed and swap four frozen backbones, including a dimension-matched random projection. The scalar stays in a narrow band while accuracy differs by up to 63.30 percentage points.
\item We measure the initialization gradient variance from 4 to 16 qubits and test entangling controls at a fixed parameter count. We condense all measurements into a protocol that pairs each diagnostic with one question.
\end{itemize}

The rest of the paper follows the order of the argument, moving from theory to measurement to interpretation. We review related work in Section~\ref{sec:related} and recall the hybrid model and the pure-state QFIM in Section~\ref{sec:prelim}. Section~\ref{sec:setup} defines the logged scalar and proves its relation to the QFIM, while the experimental setup follows in Section~\ref{sec:design}. The measurements appear in Section~\ref{sec:results}, their interpretation in Section~\ref{sec:discussion}, the scope of the study in Section~\ref{sec:future}, and the conclusion in Section~\ref{sec:conclusion}.

\begin{figure}[!h]
\centering
\includegraphics[width=\textwidth]{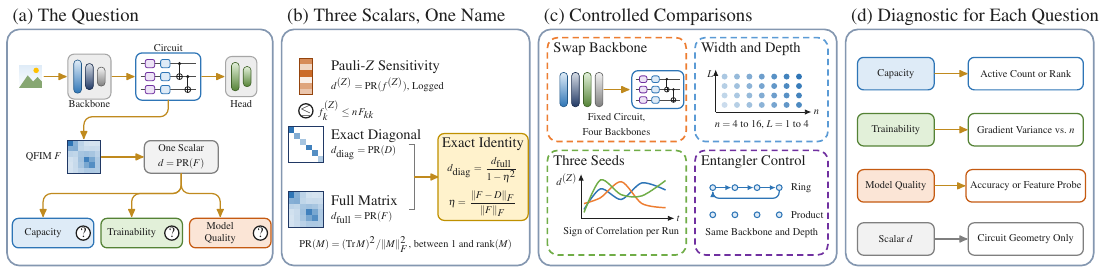}
\caption{Overview of the study. (a)~A hybrid classifier and the three roles commonly assigned to one scalar $d$ drawn from the QFIM $F$. (b)~Three scalars that share the name effective dimension: the logged Pauli-$Z$ scalar $\dz$, whose sensitivities satisfy $\fz_k\le nF_{kk}$, the diagonal participation ratio $\ddiag$, and the full participation ratio $\dfull$, with the identity that links the last two through the off-diagonal mass $\eta$. (c)~The four controlled comparisons of the study. (d)~The quantity used for each question in the protocol of Section~\ref{sec:discussion}.}
\label{fig:overview}
\end{figure}

\section{Related Work}
\label{sec:related}
In this section, we first review the QFIM as a geometric object, then contrast the scalars that share the name effective dimension, then summarize work on barren plateaus and trainability, and finally cover hybrid image classifiers built on frozen backbones.

\subsection{The Quantum Fisher Information as a Local Metric}
\label{sec:related-qfim}
Provost and Vall\'ee derived the Riemannian structure of the manifold of pure states, and the QFIM is that metric pulled back to circuit parameters \cite{provost1980}. Meyer reviewed its role in near-term quantum applications, and Liu and colleagues treated the same matrix in multiparameter estimation \cite{meyer2021,liu2020}. Stokes and colleagues used it as the preconditioner of the quantum natural gradient, a quantum analogue of the classical natural gradient \cite{stokes2020,amari1998}. Parameter-shift rules made the required derivatives available on hardware \cite{mitarai2018,wierichs2022}. Haug, Bharti, and Kim related the rank of the quantum metric to the expressivity of parameterized circuits, and Larocca and colleagues related the same rank to overparameterization \cite{haug2021,larocca2023}. These studies established $F$ as worth measuring. They left open, however, what a cheap scalar logged during training reports about $F$ once the circuit sits inside a hybrid image classifier. We relate the logged scalar to $F$ analytically, and we never use $F$ to precondition a training step.

\subsection{Two Notions of Effective Dimension}
\label{sec:related-ed}
Abbas and colleagues introduced an effective dimension to compare the capacity of quantum and classical models. Their quantity is built from the classical Fisher information, integrated over a region of parameter space, and it depends on a prior and on the number of data points. It was designed to support a generalization bound. It relates in spirit to bounds for quantum models trained on few data \cite{caro2022}, to an information-theoretic view of generalization \cite{banchi2021}, and to the effective dimensionality of trained classical networks \cite{maddox2020}. Haug and Kim related a QFIM averaged over training inputs to the generalization of quantum models \cite{haug2024}. The scalars in this paper differ in kind, because they are local in the parameters and are logged as training observables rather than as capacity scores. Prior work also left implicit how a diagonal estimate built from measured outputs relates to the matrix it is meant to summarize. We therefore import no generalization statement from the global definition, and we derive that relation directly.

\subsection{Barren Plateaus and Trainability}
\label{sec:related-bp}
A BP is the exponential decay of the cost-gradient variance as the number of qubits grows, which leaves random initializations with gradients that are indistinguishable from zero \cite{mcclean2018}. Cerezo and colleagues showed that the effect depends on the locality of the cost function in shallow circuits \cite{cerezo2021bp}. Holmes and colleagues tied it to ansatz expressibility, Wang and colleagues to hardware noise, and Sim and colleagues characterized the expressibility of common layered ansatzes \cite{holmes2022,wang2021,sim2019}. Mitigations include architectures that avoid the effect and structured initialization \cite{pesah2021,grant2019}, and the effect persists for gradient-free optimizers \cite{arrasmith2021}. Across this literature, the defining quantity is a variance and its scaling with system size. We therefore measure that variance across widths instead of reading trainability from any geometric scalar.

\subsection{Hybrid Image Classifiers}
\label{sec:related-hybrid}
Circuit-centric classifiers encode classical features into a small register and train the circuit parameters \cite{schuld2020circuit,schuld2021}. They build on the view of quantum models as feature maps and kernels in Hilbert space \cite{havlicek2019,schuld2019}, and their encoding strongly shapes what they can express \cite{schuld2021encoding,perezsalinas2020}. On images, the qubit count is far smaller than the pixel count. A frozen backbone pretrained on ImageNet therefore supplies the features \cite{deng2009,efficientnet,resnet,swin,dosovitskiy2021} in a design known as quantum transfer learning \cite{mari2020}. This design carries a known risk, since a strong backbone can mask a weak circuit \cite{cerezo2022}. We build the feature comparison around that risk. We also include a random projection matched in dimension to one learned backbone, which serves as a control without learned structure.

\section{Background}
\label{sec:prelim}
In this section, we first define the hybrid model and its trainable parts, then describe the circuit and its three entangling patterns, and finally recall the pure-state QFIM and the meaning of its diagonal.

\subsection{Hybrid Model}
\label{sec:model}
Each hybrid QNN maps an image $x$ to class scores through a frozen backbone, a trained bottleneck, a circuit, and a trained linear head. We write the composition as
\begin{equation}
\label{eq:model}
z(x) = W_{\mathrm{out}}\, o\big(\theta, a(x)\big) + c_{\mathrm{out}},
\qquad
a(x) = \tanh\!\big(W_{\mathrm{in}}\, B(x) + c_{\mathrm{in}}\big),
\end{equation}
where $B(x)\in\R^{d_B}$ is the frozen backbone feature vector, $W_{\mathrm{in}}\in\R^{n\times d_B}$ and $c_{\mathrm{in}}\in\R^{n}$ form the bottleneck, and $a(x)\in(-1,1)^{n}$ holds the angles that enter a circuit on $n$ qubits. The vector $o(\theta,a)\in[-1,1]^{n}$ collects the circuit outputs, $\theta\in\R^{p}$ holds the circuit parameters, and $W_{\mathrm{out}}$ and $c_{\mathrm{out}}$ form the head. Vectors are columns, and matrices are written in upper case. Training updates $W_{\mathrm{in}}$, $c_{\mathrm{in}}$, $\theta$, $W_{\mathrm{out}}$, and $c_{\mathrm{out}}$, so test accuracy is a property of the whole composition. In contrast, any quantity computed from the circuit depends on the backbone only through the $n$ angles that pass the bottleneck. Every feature comparison in this paper rests on that asymmetry.

\subsection{Circuit and Entangling Patterns}
\label{sec:circuit}
The circuit first encodes the angles and then applies $L$ trainable layers \cite{mitarai2018,benedetti2019}. Qubit $q$ starts in $|0\rangle$ and receives the rotation $R_y(\pi a_q)$, and each layer applies $R_z R_y R_z$ on every qubit followed by an entangling pattern \cite{schuld2020circuit,sim2019}. With three angles per qubit per layer, the parameter count is
\begin{equation}
\label{eq:p}
p = 3nL,
\end{equation}
which does not depend on the entangling pattern. We compare three patterns that differ only in their controlled-NOT (CNOT) gates. The ring applies a CNOT gate from qubit $q$ to qubit $(q+1)\bmod n$ for every $q$. The linear pattern applies the same gates on the open chain of $n-1$ edges. The product pattern applies no entangling gate at all, so the qubits stay independent until measurement. After the last layer, we measure the Pauli-$Z$ operator on each qubit,
\begin{equation}
\label{eq:expect}
o_q(\theta,a) = \langle \psi(\theta,a)|\, Z_q\, |\psi(\theta,a)\rangle, \qquad q = 1,\dots,n,
\end{equation}
where $|\psi(\theta,a)\rangle$ is the state that the circuit prepares. In the product pattern, each $o_q$ depends only on $a_q$ and on the angles that act on qubit $q$, so the circuit reduces to $n$ independent one-qubit maps.

\subsection{Pure-State Quantum Fisher Information}
\label{sec:qfim}
The QFIM measures how fast the prepared state moves when the parameters change. For a pure state, the quantum geometric tensor is
\begin{equation}
\label{eq:qgt}
\mathcal{G}_{ij} = \langle \partial_i\psi|\partial_j\psi\rangle - \langle \partial_i\psi|\psi\rangle\langle\psi|\partial_j\psi\rangle,
\end{equation}
where $\partial_i$ denotes the derivative with respect to $\theta_i$, and the QFIM is $F_{ij}=4\,\mathrm{Re}\,\mathcal{G}_{ij}$ \cite{provost1980,meyer2021}. The second term removes the component of the derivative along the state itself, which is a global phase. The matrix $F$ is real, symmetric, and positive semidefinite, and its rank counts locally independent parameter directions \cite{haug2021,larocca2023}. For a parameter that enters through $\exp(-i\theta_k \sigma_k/2)$ with a Pauli string $\sigma_k$, substituting the derivative into Eq.~\eqref{eq:qgt} gives
\begin{equation}
\label{eq:diag}
F_{kk} = 1 - \langle \psi_k|\sigma_k|\psi_k\rangle^{2} \in [0,1],
\end{equation}
where $|\psi_k\rangle$ is the state just before that gate \cite{meyer2021,stokes2020}. Each diagonal entry is therefore the variance of one generator, while the off-diagonal entries couple distinct parameters. In this paper, $F$ serves as the analytical reference against which we read the logged scalar.

\section{Measurement Framework}
\label{sec:setup}
In this section, we first define the PR and the logged Pauli-$Z$ scalar, then relate that scalar to the QFIM through a bound, a structural ceiling, and an exact identity for diagonal summaries, and finally state the trainability criterion and the research questions.

\subsection{Participation Ratio and the Logged Scalar}
\label{sec:scalars}
We summarize a nonnegative spectrum or vector by its PR, which counts how many entries carry comparable weight. For a vector $v\in\R^{p}$ with nonnegative entries that are not all zero, and for a positive semidefinite matrix $M$ with eigenvalues $\lambda_1,\dots,\lambda_p$, we define
\begin{equation}
\label{eq:pr}
\PR(v) = \frac{\big(\sum_{k} v_k\big)^{2}}{\sum_{k} v_k^{2}},
\qquad
\PR(M) = \frac{\big(\sum_{k}\lambda_k\big)^{2}}{\sum_{k}\lambda_k^{2}} = \frac{(\Tr M)^{2}}{\|M\|_F^{2}},
\end{equation}
where $\|\cdot\|_F$ is the Frobenius norm. The scalar that we log during training is built from the Pauli-$Z$ outputs of Eq.~\eqref{eq:expect}. For each parameter $k$, we shift $\theta_k$ by $\pm\pi/2$, take half the difference of the outputs, and square and sum the result over qubits,
\begin{equation}
\label{eq:fz}
\fz_k = \frac{1}{4}\,\E_x\sum_{q=1}^{n}\Big(o_q\big(\theta+\tfrac{\pi}{2}e_k,a(x)\big)-o_q\big(\theta-\tfrac{\pi}{2}e_k,a(x)\big)\Big)^{2} = \E_x\sum_{q=1}^{n}\big(\partial_k o_q\big)^{2},
\end{equation}
where $e_k$ is the $k$th unit vector and the expectation runs over the inputs of the logging batch. The second equality is the parameter-shift rule for Pauli rotations \cite{mitarai2018,wierichs2022}. We then define the logged scalar as $\dz=\PR(\fz)$ and the active count $r_Z$ as the number of entries with $\fz_k>10^{-8}$.

Two reference scalars complete the picture, and both are computed from the QFIM itself. The diagonal scalar is $\ddiag=\PR(D)$ with $D=\mathrm{diag}(F)$, and the full scalar is $\dfull=\PR(F)$. All three scalars appear in the literature under the name effective dimension, and Table~\ref{tab:notation} lists every symbol that the paper uses. The ideas behind them are illustrated in Fig.~\ref{fig:concepts}.

\begin{figure}[!h]
\centering
\includegraphics[width=\textwidth]{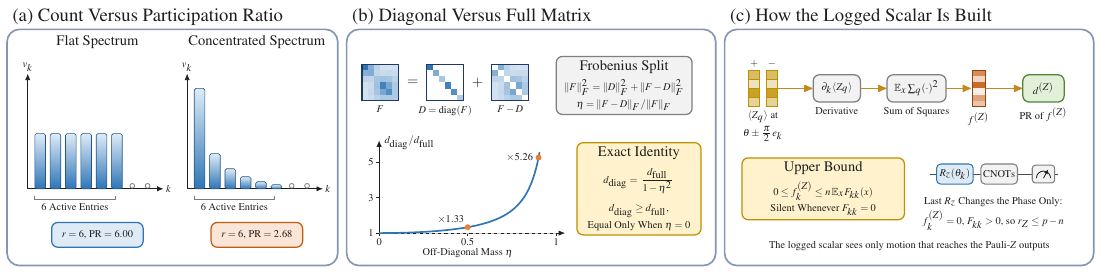}
\caption{Concepts behind the logged scalar. (a)~Two spectra with six active entries each: a flat spectrum with participation ratio 6.00 and a concentrated one with participation ratio 2.68. (b)~Decomposition of $F$ into its diagonal $D$ and off-diagonal part $F-D$, and the ratio $\ddiag/\dfull=1/(1-\eta^{2})$ against the off-diagonal mass $\eta$, with the values at $\eta=0.5$ and $\eta=0.9$ marked. (c)~Construction of $\fz$ and $\dz$ from shifted Pauli-$Z$ outputs, the bound of Proposition~\ref{prop:bound}, and the phase-only gate behind Corollary~\ref{cor:ceiling}.}
\label{fig:concepts}
\end{figure}

\begin{table}[!h]
\centering
\caption{Notation for the geometric quantities. Every quantity is computed from the circuit of Section~\ref{sec:circuit}; $p=3nL$ is the parameter count and $n$ is the number of qubits.}
\label{tab:notation}
\begin{tabular}{>{\raggedright\arraybackslash}p{2.6cm}>{\raggedright\arraybackslash}p{11.6cm}}
\toprule
Symbol & Meaning \\
\midrule
$F$, $D$ & Pure-state QFIM of Section~\ref{sec:qfim} and its diagonal part $D=\mathrm{diag}(F)$ \\
$\fz$ & Pauli-$Z$ sensitivity vector of Eq.~\eqref{eq:fz} \\
$\dz$, $r_Z$ & Logged scalar $\PR(\fz)$ and number of entries of $\fz$ above $10^{-8}$ \\
$\ddiag$, $\dfull$ & Reference scalars $\PR(D)$ and $\PR(F)$ \\
$\eta$ & Off-diagonal mass $\|F-D\|_F/\|F\|_F$ \\
$\mathrm{Var}[\partial_\theta\mathcal{L}]$ & Unbiased variance of the $p$ circuit-parameter gradients of the loss at initialization \\
\bottomrule
\end{tabular}
\end{table}

\subsection{Basic Properties}
\label{sec:properties}
The scalars share the same algebra, which fixes how they can be read before any data are collected. The following statement holds for any nonnegative vector and needs no assumption on the circuit. It therefore applies to $\fz$, to the diagonal of $F$, and to the eigenvalues of $F$ alike.

\begin{proposition}[Range, count bound, and scale invariance]
\label{prop:range}
Let $v\in\R^{p}$ have nonnegative entries that are not all zero, and let $s$ be the number of positive entries. Then $1\le \PR(v)\le s\le p$. The value $\PR(v)=s$ holds if and only if all positive entries are equal, and $\PR(v)=1$ holds if and only if $s=1$. For any $\alpha>0$, $\PR(\alpha v)=\PR(v)$.
\end{proposition}

\begin{proof}
Let $S$ be the set of indices with $v_k>0$, so $|S|=s\ge 1$. Because every entry is nonnegative,
\[
\Big(\sum_k v_k\Big)^{2} = \sum_k v_k^{2} + 2\sum_{i<j} v_i v_j \ge \sum_k v_k^{2},
\]
so $\PR(v)\ge 1$, with equality if and only if every cross term vanishes, which happens exactly when $s=1$. The Cauchy--Schwarz inequality applied to $(v_k)_{k\in S}$ and the all-ones vector on $S$ gives $(\sum_{k\in S} v_k)^{2}\le s\sum_{k\in S} v_k^{2}$, so $\PR(v)\le s$, with equality if and only if the entries indexed by $S$ are equal. The bound $s\le p$ is immediate. Replacing $v$ by $\alpha v$ multiplies the numerator and the denominator of Eq.~\eqref{eq:pr} by $\alpha^{2}$, which leaves the ratio unchanged.
\end{proof}

Proposition~\ref{prop:range} shows that each scalar counts comparably weighted directions and never exceeds its own count, so $\dz\le r_Z$ and $\dfull\le\mathrm{rank}(F)$. Scale invariance matters because the overall size of the sensitivities differs across backbones and depths, and a PR reports the shape of a spectrum rather than its size. As shown in Fig.~\ref{fig:concepts}(a), two spectra with six active entries each can have PRs of 6.00 and 2.68.

\subsection{Relation of the Logged Scalar to the QFIM}
\label{sec:bound}
The Pauli-$Z$ sensitivity measures how fast the outputs move, while $F_{kk}$ measures how fast the state moves. Proposition~\ref{prop:bound} shows that the first can never exceed the second by more than a factor of $n$.

\begin{proposition}[Output sensitivity and the QFIM diagonal]
\label{prop:bound}
For every input $x$, every parameter $k$, and every qubit $q$, $(\partial_k o_q)^{2}\le F_{kk}\,\mathrm{Var}(Z_q)$, where the variance is taken in the state $|\psi(\theta,a(x))\rangle$. Consequently,
\begin{equation}
\label{eq:bound}
0\le \fz_k \le \E_x\Big[F_{kk}(x)\sum_{q=1}^{n}\mathrm{Var}(Z_q)\Big] \le n\,\E_x F_{kk}(x),
\end{equation}
and $\fz_k=0$ whenever $F_{kk}(x)=0$ for every input of the batch.
\end{proposition}

\begin{proof}
Write $|\psi\rangle$ for the state at input $x$ and $\Delta Z_q=Z_q-o_q$. Since $\langle\psi|\psi\rangle=1$, we have $\partial_k o_q=2\,\mathrm{Re}\langle\partial_k\psi|\Delta Z_q|\psi\rangle$. The vector $\Delta Z_q|\psi\rangle$ is orthogonal to $|\psi\rangle$, so $|\partial_k\psi\rangle$ can be replaced by its component $|\partial_k\psi_\perp\rangle$ orthogonal to $|\psi\rangle$. The Cauchy--Schwarz inequality then gives $|\partial_k o_q|\le 2\,\|\partial_k\psi_\perp\|\,\|\Delta Z_q\psi\|$. By Eq.~\eqref{eq:qgt}, $4\|\partial_k\psi_\perp\|^{2}=F_{kk}$, and $\|\Delta Z_q\psi\|^{2}=\mathrm{Var}(Z_q)\le 1$. Squaring, summing over $q$, and averaging over $x$ gives Eq.~\eqref{eq:bound}.
\end{proof}

The bound is one-sided, because a parameter can move the state without moving any Pauli-$Z$ output. The circuit of Section~\ref{sec:circuit} contains $n$ such parameters by construction.

\begin{corollary}[Structural ceiling on the active count]
\label{cor:ceiling}
For the ring, linear, and product patterns, the final $R_z$ angle on each qubit in the last layer has $\fz_k=0$ exactly, so $r_Z\le p-n$ and $r_Z/p\le 1-1/(3L)$.
\end{corollary}

\begin{proof}
In the last layer, the final $R_z$ gates are diagonal in the computational basis, and the only gates that follow them are CNOT gates, which permute computational basis states. A diagonal phase followed by a permutation leaves every computational-basis probability unchanged, so every $o_q$ is independent of these $n$ angles and $\fz_k=0$ for each of them.
\end{proof}

For these same $n$ angles, $F_{kk}$ is generally positive by Eq.~\eqref{eq:diag}, so the QFIM sees motion that the logged scalar cannot. The ceiling $1-1/(3L)$ equals $0.67$, $0.83$, $0.89$, and $0.92$ for one to four layers, and Section~\ref{sec:rq2} compares it with the measured fill.

\subsection{Diagonal Summaries and the Full Matrix}
\label{sec:identity}
Any scalar built from a diagonal ignores the couplings between parameters. Whether a diagonal PR can stand in for $\dfull$ depends on how much of the Frobenius mass of $F$ sits off the diagonal, which we measure by
\begin{equation}
\label{eq:eta}
\eta = \frac{\|F-D\|_F}{\|F\|_F}.
\end{equation}
A diagonal matrix has $\eta=0$, and a matrix whose mass sits mostly off the diagonal has $\eta$ close to one. The next statement turns this mass into an exact correction.

\begin{proposition}[Diagonal and full PRs]
\label{prop:identity}
Let $F\neq 0$ be real, symmetric, and positive semidefinite, let $D=\mathrm{diag}(F)$, and let $\eta$ be as in Eq.~\eqref{eq:eta}. Then $\eta<1$ and
\begin{equation}
\label{eq:identity}
\ddiag = \frac{\dfull}{1-\eta^{2}}.
\end{equation}
\end{proposition}

\begin{proof}
The matrices $D$ and $F-D$ have disjoint supports, so their Frobenius inner product vanishes and $\|F\|_F^{2}=\|D\|_F^{2}+\|F-D\|_F^{2}$. Dividing by $\|F\|_F^{2}$ gives $\|D\|_F^{2}=(1-\eta^{2})\|F\|_F^{2}$. Since $F\neq 0$ is positive semidefinite, $\Tr F>0$, so $D\neq 0$ and $\eta<1$. The trace depends only on the diagonal, so $\Tr D=\Tr F$, and therefore $\ddiag=(\Tr F)^{2}/\|D\|_F^{2}=(\Tr F)^{2}/\big((1-\eta^{2})\|F\|_F^{2}\big)=\dfull/(1-\eta^{2})$.
\end{proof}

Proposition~\ref{prop:identity} shows that a diagonal PR always overstates the full one, and that the overstatement depends on $\eta$ alone. At $\eta=1/2$ the ratio is $4/3$, and at $\eta=0.9$ it already exceeds five, as shown in Fig.~\ref{fig:concepts}(b). Together with Proposition~\ref{prop:bound}, it states what a diagonal summary can and cannot say about the spectrum of $F$ without computing the matrix.

\subsection{Trainability Criterion}
\label{sec:trainability}
Trainability is a statement about the variance of the cost gradient as the system grows \cite{mcclean2018,cerezo2021bp}. We therefore log $\mathrm{Var}[\partial_\theta\mathcal{L}]$, the unbiased variance of the $p$ circuit-parameter gradients. We compute this variance from one backward pass of the cross-entropy loss $\mathcal{L}$ on 256 cached training rows at initialization. We fixed a plateau rule before reading the grid. Under this rule, a width counts as a plateau only when the variance lies below $10^{-7}$ and decreases across two successive widths. A threshold rule can miss an early trend. We therefore also fit $\log_{10}\mathrm{Var}[\partial_\theta\mathcal{L}]$ against $n$ and against $\log_{10} n$.

\subsection{Research Questions}
\label{sec:rqs}
We organize the grid so that each research question (RQ) moves one factor while the others stay fixed.
\begin{itemize}
\item RQ1: How does $\dz$ evolve during training, and does it track test accuracy consistently across seeds?
\item RQ2: How do $r_Z$ and $\dz$ change with depth, and how close does $r_Z$ come to the ceiling of Corollary~\ref{cor:ceiling}?
\item RQ3: How does the initialization gradient variance change with depth and with width?
\item RQ4: At a fixed circuit, does $\dz$ order the backbones that accuracy orders?
\item RQ5: How does accuracy change with width and with the entangling pattern at a fixed parameter count?
\end{itemize}

\section{Experimental Setup}
\label{sec:design}
In this section, we first describe the datasets and the frozen backbones, then the circuit and the measurement grid, and finally the training protocol, the logging schedule, and the statistics we report.

\begin{figure}[!h]
\centering
\includegraphics[width=\textwidth]{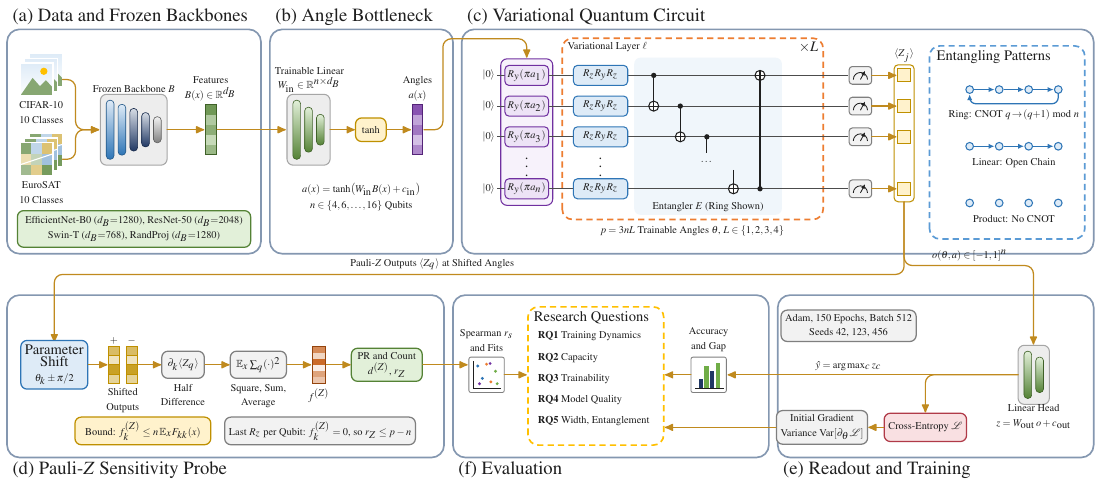}
\caption{Measurement pipeline. (a)~Images from CIFAR-10 and EuroSAT pass through one of four frozen backbones. (b)~A trained linear map and $\tanh$ produce $n$ angles $a(x)$. (c)~The circuit encodes the angles with $R_y(\pi a_q)$ and applies $L$ layers of $R_zR_yR_z$ rotations and an entangling pattern; the ring is drawn, and the linear and product patterns are listed on the right. (d)~Pauli-$Z$ outputs at shifted angles give the sensitivities $\fz$, the logged scalar $\dz$, and the active count $r_Z$, with the bound and the ceiling of Section~\ref{sec:bound}. (e)~The outputs feed a linear head trained with cross-entropy and Adam, and the gradient variance is taken at initialization. (f)~Accuracy, the train--test gap, and the logged scalars enter the five research questions.}
\label{fig:method}
\end{figure}

\subsection{Datasets}
\label{sec:datasets}
We use CIFAR-10 \cite{cifar} and EuroSAT \cite{eurosat}, two ten-class image collections with different visual structure. CIFAR-10 contains natural object images, while EuroSAT contains land-cover tiles from satellite imagery. We therefore call a finding consistent when it holds on both datasets, and dataset dependent when it holds on only one. Both loaders apply channel-wise normalization with the mean and standard deviation (std) of each dataset before any feature map. Backbone features were extracted once and cached, and later epochs read the cache. The training images passed through one augmentation draw during that extraction. That draw used random crop and horizontal flip for CIFAR-10 and horizontal flip and color jitter for EuroSAT, and test images were not augmented.

\subsection{Backbones and Controls}
\label{sec:backbones}
We used four frozen feature maps, each differing only in the vector $B(x)$ passed to the bottleneck, as shown in Fig.~\ref{fig:method}(a). EfficientNet-B0 maps a $96\times 96$ image to $1280$ features \cite{efficientnet}, and ResNet-50 maps a $224\times 224$ image to $2048$ features \cite{resnet}. Swin-T maps a $224\times 224$ image to $768$ features \cite{swin,dosovitskiy2021}. The three learned backbones carry weights pretrained on ImageNet \cite{deng2009}. The fourth map, which we call RandProj, is a fixed Gaussian random projection of the flattened normalized image to $1280$ features. It sees the native $32\times 32$ image on CIFAR-10 and the $96\times 96$ resized image on EuroSAT. We chose these maps because they span a wide accuracy range at a fixed circuit, and RandProj matches the feature dimension of EfficientNet-B0 exactly.

The study measures a diagnostic rather than proposing a model, so it has controls instead of competing baselines. RandProj is the control without learned visual structure, which blocks the explanation that the learned backbones merely embed images alike. The product pattern is the control without entanglement, since it keeps the parameter count of Eq.~\eqref{eq:p} and removes every CNOT gate.

\subsection{Circuit and Initialization}
\label{sec:init}
The circuit follows Section~\ref{sec:circuit}, with the ring as the default pattern. We drew the circuit parameters from $\mathcal{N}(0,0.1^{2})$, and the bottleneck and the head used the default initialization of linear layers in PyTorch \cite{pytorch}. Width $n$ ranged over $\{4,6,\dots,16\}$ and depth $L$ over $\{1,2,3,4\}$, which gives $p$ from $12$ to $192$.

\subsection{Measurement Grid}
\label{sec:grid}
The study uses four sub-grids, listed in Table~\ref{tab:grid}, and each comparison draws on one of them. Sub-grid A fixes the register at eight qubits and logs $\dz$ along full training runs, while sub-grid B measures the gradient variance at initialization across width and depth. We trained sub-grid C for the width, backbone, and entanglement comparisons, and we varied the number of ring CNOT gates in sub-grid D at a fixed parameter count. Width $16$ was measured at initialization only, because a training attempt at $n=16$ exhausted GPU memory.

\begin{table}[!h]
\centering
\caption{Measurement grid. Every configuration was run on CIFAR-10 and EuroSAT with seeds 42, 123, and 456. EffNet denotes EfficientNet-B0, and a cell is one configuration, dataset, and seed.}
\label{tab:grid}
\resizebox{\linewidth}{!}{\begin{tabular}{>{\raggedright\arraybackslash}p{1.6cm}>{\raggedright\arraybackslash}p{8.8cm}>{\raggedright\arraybackslash}p{4.2cm}r}
\toprule
Sub-grid & Configurations & Logged quantities & Cells \\
\midrule
A & $n=8$; four backbones at $L=2$, and EffNet at $L\in\{1,2,3,4\}$; ring & $\fz$, $\dz$, $r_Z$, gradient variance, accuracy & 42 \\
\midrule
B & EffNet; $n\in\{4,6,\dots,16\}$; $L\in\{1,2,3,4\}$; ring; untrained & Gradient variance at initialization & 168 \\
\midrule
C & EffNet ring at $L=2$, $n\in\{4,\dots,12\}$; three other backbones, the linear and product patterns, and $L\in\{1,3,4\}$ at $n\in\{8,12\}$ & Accuracy and train--test gap & 126 \\
\midrule
D & EffNet at $L=2$; first $k$ ring CNOT gates, $k\in\{2,4,6\}$ at $n=8$ and $k\in\{3,6,9\}$ at $n=12$ & Accuracy & 36 \\
\bottomrule
\end{tabular}}
\end{table}

\subsection{Training Protocol}
\label{sec:training}
Every trained configuration ran for $150$ epochs of Adam \cite{adam} with cross-entropy loss and batch size $512$. The bottleneck and head used learning rate $10^{-3}$, and the circuit parameters used learning rate $5\times 10^{-2}$. All runs used a noise-free statevector simulator that we wrote in PyTorch \cite{pytorch}, with single-precision complex amplitudes and automatic differentiation through the circuit. We report the best test accuracy along the trajectory together with the final test accuracy. The train--test gap is the final training accuracy minus the final test accuracy. No run used the quantum natural gradient, early stopping on a geometric scalar, or a classical replacement of the circuit.

\subsection{Logging and Statistics}
\label{sec:stats}
On sub-grid A, we logged $\fz$, $\dz$, $r_Z$, and the gradient variance once at initialization. We logged them again after the first epoch and every fifth epoch, which gave 31 training checkpoints per run. On sub-grid B, we logged the gradient variance at initialization as defined in Section~\ref{sec:trainability}. Continuous summaries are the mean and std over the three seeds, written as the mean followed by the std in small type. Associations across backbones are Spearman rank correlations $r_s$, and the per-run association between two time series is the Pearson correlation $\rho$. Backbone accuracy gaps are tested with paired $t$-tests over seeds, where $P$ denotes a two-sided $P$-value. We report the Wilcoxon signed-rank test only as a sign check, because with three seeds its $P$-value cannot fall below $0.25$. Least-squares fits use the seed means and report the standard error of each slope.

\section{Results}
\label{sec:results}
In this section, we first follow the logged scalar through training, then report capacity and trainability as depth and width grow, then test whether the scalar follows backbone quality, and finally examine how width and entanglement change accuracy.

\subsection{The Logged Scalar During Training}
\label{sec:rq1}
During training, $\dz$ moves within a band far below $p$, and its association with accuracy is not stable across seeds (RQ1). At $n=8$ and $L=2$, the accuracy curves separate early by backbone and keep that separation, while the geometry of every backbone stays within a shared band, as shown in Fig.~\ref{fig:traj}. RandProj on CIFAR-10 is the only trajectory that rises well above the others, and its geometry also settles last, as reported in Supplementary Table~\ref{tab:dyn}.

The per-run Pearson correlation $\rho$ between the $\dz$ series and the accuracy series has a positive mean in 13 of the 14 configurations, as reported in Table~\ref{tab:corr}. However, seven configurations contain seeds of opposite sign, and single runs range from $-0.72$ to $+0.87$, as shown in Fig.~\ref{fig:corr}. All four configurations with $L\ge 3$ are sign-stable on our three seeds, but stability measured on three seeds does not guarantee the sign for a new seed. Much of the positive mean likely reflects the shared early phase of training, in which both quantities change fastest.

\begin{figure}[!h]
\centering
\includegraphics[width=0.7\textwidth]{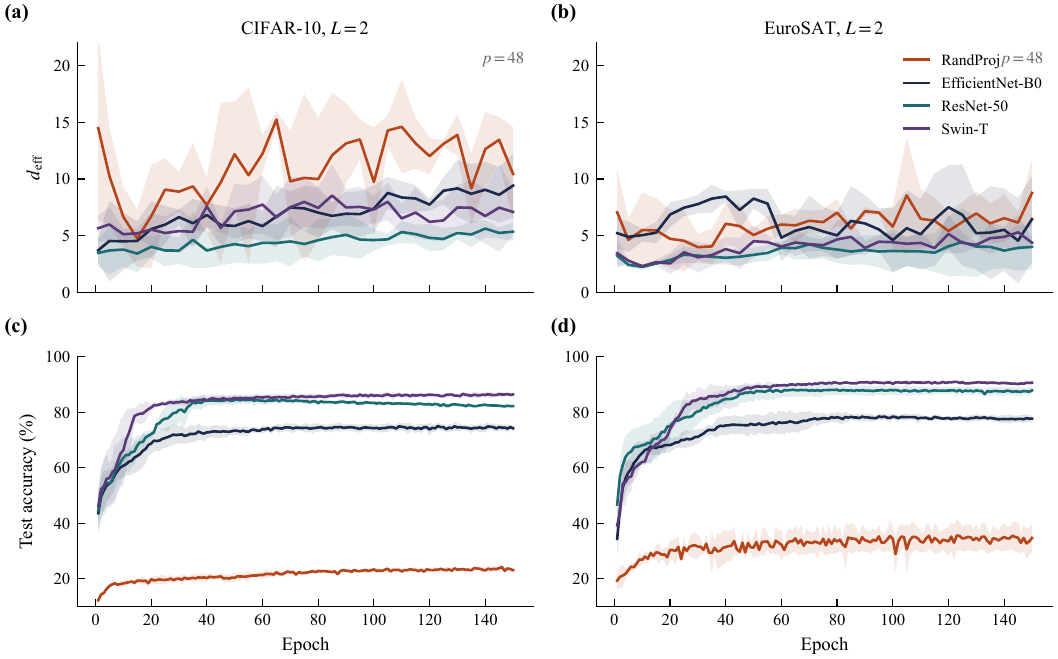}
\caption{Training dynamics at $n=8$ and $L=2$ (sub-grid A), with bands of one std over three seeds. (a,b)~The logged scalar $\dz$, labeled $d_{\mathrm{eff}}$ on the axis, against epoch on CIFAR-10 and EuroSAT. (c,d)~Test accuracy against epoch on the same runs.}
\label{fig:traj}
\end{figure}

\begin{table}[!h]
\centering
\caption{Per-run Pearson correlation $\rho$ between the $\dz$ and test-accuracy time series over the 14 configurations of sub-grid A. A configuration is sign-stable only when all three seeds give a positive $\rho$.}
\label{tab:corr}
\begin{tabular}{lccc}
\toprule
Statistic & CIFAR-10 & EuroSAT & All \\
\midrule
Configurations with positive mean & 7/7 & 6/7 & 13/14 \\
Sign-stable configurations & 3/7 & 4/7 & 7/14 \\
Configurations whose range crosses zero & 4/7 & 3/7 & 7/14 \\
Sign-stable configurations with $L\ge 3$ & 2/2 & 2/2 & 4/4 \\
Mean of configuration means & $+0.27$ & $+0.25$ & $+0.26$ \\
Lowest and highest single-run $\rho$ & \multicolumn{3}{c}{$-0.72$ and $+0.87$} \\
\bottomrule
\end{tabular}
\end{table}

\begin{figure}[!h]
\centering
\includegraphics[width=0.66\textwidth]{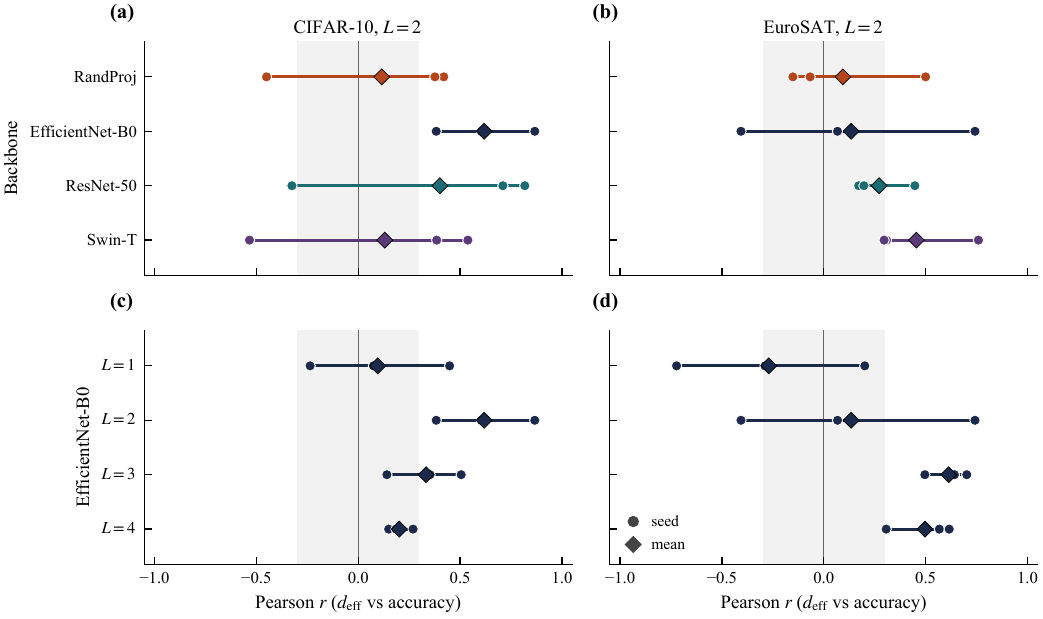}
\caption{Per-seed Pearson correlation $\rho$ between $\dz$ and test accuracy (sub-grid A). Small markers are seeds and large markers are seed means. (a,b)~Four backbones at $L=2$. (c,d)~EfficientNet-B0 at $L\in\{1,2,3,4\}$.}
\label{fig:corr}
\end{figure}

\subsection{Capacity Across Depth}
\label{sec:rq2}
Depth fills the active fraction of parameters up to the structural ceiling of Corollary~\ref{cor:ceiling}, while the logged scalar grows far more slowly (RQ2). The fraction $r_Z/p$ rises from $0.39$ at one layer to $0.91$ at four layers on CIFAR-10 and from $0.38$ to $0.88$ on EuroSAT, as reported in Table~\ref{tab:depth}. With two or more layers, the measured fill sits within 0.04 of the ceiling $1-1/(3L)$ on both datasets, as shown in Fig.~\ref{fig:depth}(b). Most of the fill seen beyond one layer is therefore structural, since only the $n$ phase-only angles stay silent.

One layer is the exception, with $r_Z/p$ well below its ceiling of $0.67$. A likely reason is that a single layer gives the readout little room to register each rotation once the CNOT gates have acted. Over the same sweep, the initial value of $\dz$ grows from $5.76$ to $15.23$ on CIFAR-10, while its ratio to $p$ falls from $0.24$ to $0.16$, as reported in Table~\ref{tab:depth}. Deeper circuits therefore gain active directions without sharing the sensitivity more evenly among them. Accuracy rises only modestly, by $3.30$ points on CIFAR-10 and $6.86$ points on EuroSAT between one and four layers.

\begin{table}[!h]
\centering
\caption{Depth sweep at $n=8$ (EfficientNet-B0, ring, sub-grid A). Entries are the mean $\pm$ std over three seeds; $r_Z$ is averaged over the 31 checkpoints of each run, and the ceiling is $1-1/(3L)$ from Corollary~\ref{cor:ceiling}. Bold marks the best final accuracy and underline the second best.}
\label{tab:depth}
\resizebox{0.65\linewidth}{!}{\begin{tabular}{lccccccc}
\toprule
Dataset & $L$ & $p$ & Final acc.\ (\%) & Initial $\dz$ & Mean $r_Z$ & $r_Z/p$ & Ceiling \\
\midrule
\multirow{4}{*}{CIFAR-10}
 & 1 & 24 & 72.40\sd{3.15} & 5.76\sd{0.30} & 9.3\sd{0.1} & 0.39 & 0.67 \\
 & 2 & 48 & 74.25\sd{1.05} & 9.13\sd{0.68} & 38.2\sd{0.5} & 0.80 & 0.83 \\
 & 3 & 72 & \underline{74.79}\sd{0.72} & 11.55\sd{0.30} & 63.2\sd{0.4} & 0.88 & 0.89 \\
 & 4 & 96 & \textbf{75.70}\sd{0.55} & 15.23\sd{2.28} & 87.2\sd{0.2} & 0.91 & 0.92 \\
\midrule
\multirow{4}{*}{EuroSAT}
 & 1 & 24 & 74.12\sd{0.95} & 6.03\sd{0.71} & 9.0\sd{0.1} & 0.38 & 0.67 \\
 & 2 & 48 & 77.71\sd{1.20} & 9.28\sd{0.23} & 38.9\sd{0.2} & 0.81 & 0.83 \\
 & 3 & 72 & \underline{80.22}\sd{1.63} & 12.65\sd{1.02} & 63.4\sd{0.3} & 0.88 & 0.89 \\
 & 4 & 96 & \textbf{80.98}\sd{0.32} & 15.30\sd{1.63} & 84.7\sd{1.4} & 0.88 & 0.92 \\
\bottomrule
\end{tabular}}
\end{table}

\begin{figure}[!h]
\centering
\includegraphics[width=0.85\textwidth]{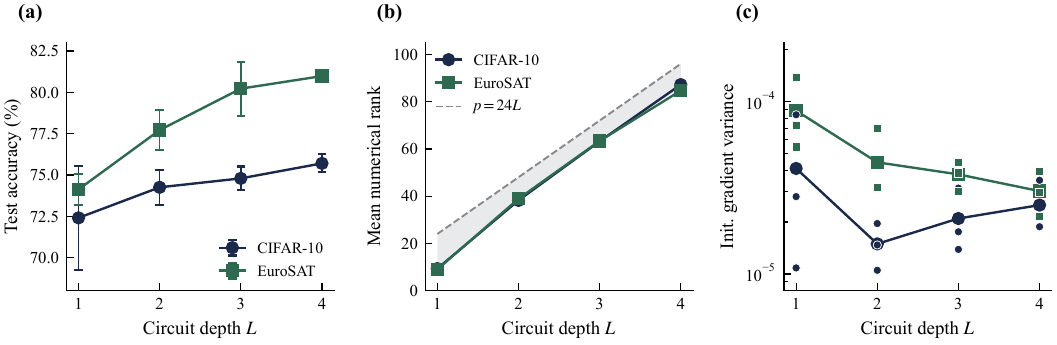}
\caption{Depth sweep at $n=8$ (EfficientNet-B0, ring, sub-grid A), as the mean and std over three seeds. (a)~Final test accuracy. (b)~Active count $r_Z$, labeled mean numerical rank on the axis; the dashed line is $p=24L$. (c)~Initialization gradient variance on a logarithmic axis.}
\label{fig:depth}
\end{figure}

\subsection{Trainability Across Depth and Width}
\label{sec:rq3}
The initialization gradient variance never meets the plateau rule, yet it decays with width at a rate that matters for larger registers (RQ3). From $n=4$ to $n=16$ at depth two, the variance falls by factors of $34.5$ on CIFAR-10 and $9.0$ on EuroSAT, as reported in Tables~\ref{tab:width} and~\ref{tab:fits} and shown in Fig.~\ref{fig:width}(a). An exponential fit gives a slope of $-0.139\pm 0.019$ decades per qubit on CIFAR-10 and $-0.092\pm 0.025$ on EuroSAT. The exponential form fits better than a power law on both datasets, with $R^{2}$ of $0.92$ against $0.83$ and $0.72$ against $0.54$. However, seven widths are too few to separate the two forms with confidence.

Depth at eight qubits shows no comparable trend in the gradient variance. The variance stays between $1.49\times10^{-5}$ and $8.86\times10^{-5}$ for $L\le 4$, and it does not decrease monotonically with depth on CIFAR-10, as shown in Fig.~\ref{fig:depth}(c). Training does not shrink it either, since the late-training variance remains of the same order as the variance at the first checkpoint, as shown in Supplementary Fig.~\ref{fig:trainvar}. We therefore read the measurements as a finite-size statement. No measured width or depth is a plateau under the fixed rule, while the width trend is the early decay that a BP analysis at larger $n$ would need to track.

\begin{table}[!h]
\centering
\caption{Initialization gradient variance against width at $L=2$ (EfficientNet-B0, ring, sub-grid B), with $p=6n$, in units of $10^{-6}$. Entries are the mean $\pm$ std over three seeds.}
\label{tab:width}
\begin{tabular}{rrcc}
\toprule
$n$ & $p$ & CIFAR-10 & EuroSAT \\
\midrule
4  & 24 & 42.78\sd{11.14} & 25.57\sd{19.23} \\
6  & 36 & 46.04\sd{15.92} & 62.11\sd{70.45} \\
8  & 48 & 17.51\sd{13.60} & 42.92\sd{22.30} \\
10 & 60 & 18.70\sd{9.82} & 30.96\sd{7.16} \\
12 & 72 & 2.95\sd{1.41} & 16.06\sd{5.95} \\
14 & 84 & 2.97\sd{1.66} & 7.26\sd{2.78} \\
16 & 96 & 1.24\sd{0.70} & 2.83\sd{0.03} \\
\bottomrule
\end{tabular}
\end{table}

\begin{table}[!h]
\centering
\caption{Least-squares fits of the initialization gradient variance against width at $L=2$ (EfficientNet-B0, ring), computed on the seed means of Table~\ref{tab:width} for $n\in\{4,6,\dots,16\}$. Slopes are given $\pm$ their standard errors.}
\label{tab:fits}
\begin{tabular}{lcc}
\toprule
Statistic & CIFAR-10 & EuroSAT \\
\midrule
Variance ratio, $n=4$ over $n=16$ & 34.5 & 9.0 \\
Exponential slope of $\log_{10}\mathrm{Var}$ on $n$ (decades per qubit) & $-0.139$\sd{0.019} & $-0.092$\sd{0.025} \\
$R^{2}$ of the exponential fit & 0.92 & 0.72 \\
Power-law exponent of $\log_{10}\mathrm{Var}$ on $\log_{10} n$ & $-2.67$\sd{0.54} & $-1.60$\sd{0.67} \\
$R^{2}$ of the power-law fit & 0.83 & 0.54 \\
\bottomrule
\end{tabular}
\end{table}

\begin{figure}[!h]
\centering
\includegraphics[width=0.92\textwidth]{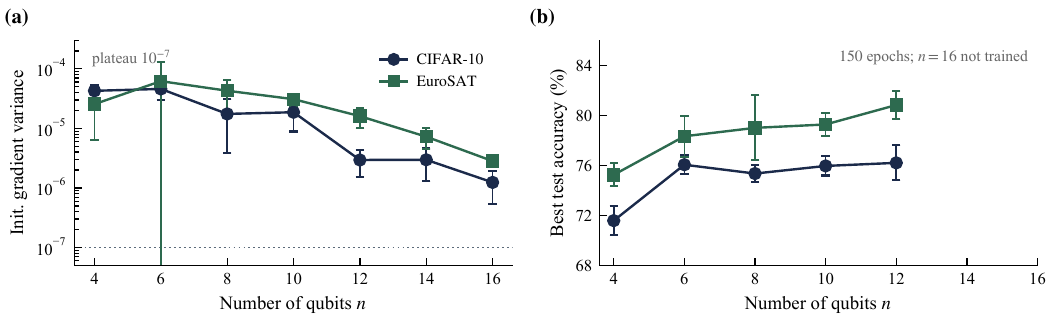}
\caption{Width sweep at $L=2$ (EfficientNet-B0, ring), as the mean and std over three seeds. (a)~Initialization gradient variance on a logarithmic axis (sub-grid B); the dotted line is the plateau level $10^{-7}$, and the lower EuroSAT bar at $n=6$ leaves the axis because the std exceeds the mean. (b)~Best test accuracy after 150 epochs (sub-grid C); $n=16$ was not trained.}
\label{fig:width}
\end{figure}

\subsection{The Logged Scalar Does Not Order the Backbones}
\label{sec:rq4}
The logged scalar does not follow the backbone order that accuracy shows (RQ4). At a fixed circuit with $n=8$ and $L=2$, CIFAR-10 accuracy rises from $23.14\%$ for RandProj to $86.44\%$ for Swin-T, as reported in Table~\ref{tab:scope}. Over the same backbones, the initial value of $\dz$ stays between $6.63$ and $9.13$. The Spearman correlation between initial $\dz$ and final accuracy is $0.00$ on both datasets, and the correlation for the peak value of $\dz$ is $-0.80$ on CIFAR-10 and $-1.00$ on EuroSAT. With four backbones these coefficients are descriptive, yet both indicate that neither the starting nor the peak geometry follows feature quality.

The accuracy order itself is robust across seeds and widths. All twelve paired accuracy gaps between backbones are significant under paired $t$-tests, as reported in Supplementary Table~\ref{tab:sig}. The order from RandProj through EfficientNet-B0 and ResNet-50 to Swin-T also holds at $n=12$, as shown in Fig.~\ref{fig:backbones}. RandProj and EfficientNet-B0 share the same feature dimension, yet their accuracies differ by more than 40 points on both datasets. A scalar that cannot separate them does not measure feature quality, which is consistent with the asymmetry of Section~\ref{sec:model}.

\begin{table}[!h]
\centering
\caption{Feature scope at a fixed circuit ($n=8$, $L=2$, $p=48$, ring, sub-grid A). Entries are the mean $\pm$ std over three seeds. Bold marks the best accuracy and underline the second best.}
\label{tab:scope}
\resizebox{0.75\linewidth}{!}{\begin{tabular}{llccccc}
\toprule
Dataset & Backbone & $d_B$ & Initial $\dz$ & Peak $\dz$ & Final acc.\ (\%) & Best acc.\ (\%) \\
\midrule
\multirow{4}{*}{CIFAR-10}
 & RandProj & 1280 & 8.46\sd{0.20} & 18.72\sd{3.03} & 23.14\sd{0.97} & 24.48\sd{0.82} \\
 & EfficientNet-B0 & 1280 & 9.13\sd{0.68} & 10.09\sd{2.19} & 74.25\sd{1.05} & 75.49\sd{1.16} \\
 & ResNet-50 & 2048 & 6.63\sd{1.17} & 6.21\sd{0.42} & \underline{82.22}\sd{0.43} & \underline{85.10}\sd{1.31} \\
 & Swin-T & 768 & 9.08\sd{1.29} & 9.78\sd{2.15} & \textbf{86.44}\sd{0.61} & \textbf{86.88}\sd{0.73} \\
\midrule
\multirow{4}{*}{EuroSAT}
 & RandProj & 1280 & 8.03\sd{1.08} & 10.80\sd{3.13} & 34.69\sd{4.86} & 36.35\sd{4.89} \\
 & EfficientNet-B0 & 1280 & 9.28\sd{0.23} & 9.79\sd{1.07} & 77.71\sd{1.20} & 78.94\sd{1.29} \\
 & ResNet-50 & 2048 & 6.18\sd{0.52} & 5.68\sd{1.98} & \underline{87.88}\sd{1.24} & \underline{88.85}\sd{1.90} \\
 & Swin-T & 768 & 8.90\sd{1.51} & 5.66\sd{1.10} & \textbf{90.69}\sd{0.39} & \textbf{91.38}\sd{0.47} \\
\bottomrule
\end{tabular}}
\end{table}

\begin{figure}[!h]
\centering
\includegraphics[width=0.45\textwidth]{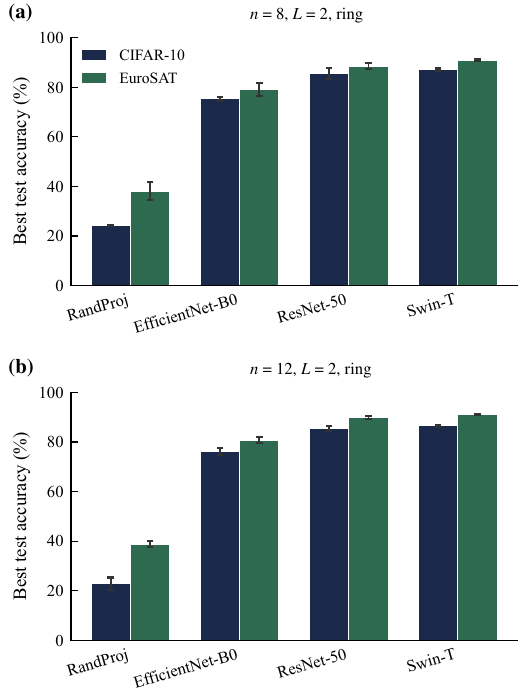}
\caption{Best test accuracy for four backbones at a fixed circuit (ring, $L=2$, sub-grid C) at (a)~$n=8$ and (b)~$n=12$. Bars are seed means and error bars are one std over three seeds.}
\label{fig:backbones}
\end{figure}

\subsection{Width and Entanglement}
\label{sec:rq5}
Accuracy saturates with width, and at a fixed parameter count, removing entanglement raises accuracy (RQ5). Best EfficientNet-B0 accuracy on CIFAR-10 is near 76\% from $n=6$ onward, as shown in Fig.~\ref{fig:width}(b). EuroSAT accuracy rises from about 75\% at $n=4$ to about 81\% at $n=12$ over the same sweep. Extra qubits therefore add parameters faster than they add accuracy, a point we return to in Section~\ref{sec:discussion}.

At $L=2$ with EfficientNet-B0, the product pattern exceeds the ring by $5.3$ to $8.2$ points across both widths and datasets, as shown in Fig.~\ref{fig:entangle}. The linear pattern, which drops a single CNOT gate from the ring, sits between them in all four comparisons. Because the parameter count is identical across the three patterns, the change in accuracy comes from the circuit structure rather than from capacity in the counting sense. The CNOT-count experiment of sub-grid D points in the same direction. Replacing the full ring by its first $k$ CNOT gates leaves $p$ unchanged, and the mean final accuracy does not rise with $k$, as reported in Table~\ref{tab:dose}. From the smallest to the largest $k$, accuracy falls in all four settings, by $0.86$ to $2.51$ points. The largest drop, on EuroSAT at $n=12$, exceeds the std of each of the three cells it compares.

\begin{figure}[!h]
\centering
\includegraphics[width=0.7\textwidth]{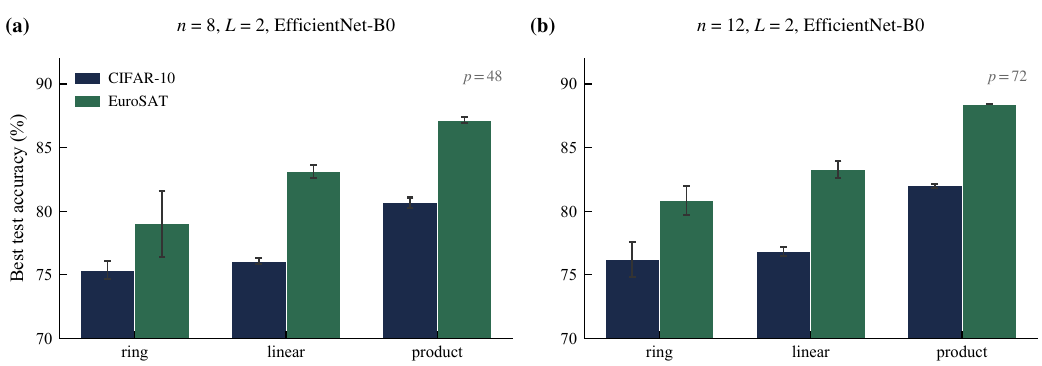}
\caption{Best test accuracy for the ring, linear, and product patterns (EfficientNet-B0, $L=2$, sub-grid C) at (a)~$n=8$ with $p=48$ and (b)~$n=12$ with $p=72$. Bars are seed means and error bars are one std over three seeds.}
\label{fig:entangle}
\end{figure}

\begin{table}[!h]
\centering
\caption{Final test accuracy (\%) when the ring is replaced by its first $k$ CNOT gates (EfficientNet-B0, $L=2$, sub-grid D). Entries are the mean $\pm$ std over three seeds. The parameter count $p=3nL$ does not depend on $k$.}
\label{tab:dose}
\begin{tabular}{llcc}
\toprule
Width & $k$ & CIFAR-10 & EuroSAT \\
\midrule
$n=8$ & 2 & 77.26\sd{0.83} & 82.54\sd{0.47} \\
$n=8$ & 4 & 76.56\sd{0.40} & 81.77\sd{0.92} \\
$n=8$ & 6 & 75.45\sd{1.34} & 80.25\sd{1.19} \\
\midrule
$n=12$ & 3 & 72.99\sd{0.47} & 83.53\sd{0.38} \\
$n=12$ & 6 & 72.91\sd{0.10} & 81.97\sd{0.72} \\
$n=12$ & 9 & 72.13\sd{0.53} & 81.02\sd{0.73} \\
\bottomrule
\end{tabular}
\end{table}

\section{Discussion}
\label{sec:discussion}
In this section, we first interpret what the logged scalar measures, then relate diagonal summaries to the full QFIM, then connect accuracy to the data path and to entanglement, and finally condense the measurements into a protocol for hybrid QNNs.

\subsection{What the Logged Scalar Measures}
\label{sec:disc-scalars}
The logged scalar is an observable of local circuit geometry as seen through the Pauli-$Z$ readout, and it is not a measure of the hybrid model. Proposition~\ref{prop:bound} places it below the QFIM diagonal, and Corollary~\ref{cor:ceiling} shows that it is blind to the $n$ phase-only directions that the QFIM still sees. Accuracy, in contrast, belongs to the whole composition of Eq.~\eqref{eq:model}, whereas the circuit sees the backbone only through $n$ angles that the bottleneck maps into a comparable range. Changing the backbone can therefore reshape the loss landscape of the composite model while leaving the local circuit geometry nearly unchanged. RandProj provides a direct test of this argument, because its accuracy differs sharply from the learned backbones while its geometry does not follow.

\subsection{Diagonal Summaries and the Full Matrix}
\label{sec:disc-diag}
Proposition~\ref{prop:identity} sets the price of reading the spectrum of $F$ from its diagonal. Any diagonal PR overstates the full one by $1/(1-\eta^{2})$, so it is faithful only when the couplings between parameters carry little of the Frobenius mass. A logged diagonal estimate such as $\dz$ adds a second gap, since it sits below the QFIM diagonal by Proposition~\ref{prop:bound}. We therefore recommend reporting the full PR and the numerical rank whenever the spectrum of $F$ is the object of interest, and quoting any diagonal scalar together with $\eta$. A full-matrix implementation also admits a simple check that the product pattern supplies. Every state of that pattern is a product of one-qubit states, which form a manifold of real dimension $2n$, so the numerical rank of a correct QFIM cannot exceed $2n$.

\subsection{The Data Path Sets Accuracy}
\label{sec:disc-data}
Accuracy saturates while the active count keeps growing, which points to the data path rather than to the circuit. The bottleneck has $n$ outputs, and once $n$ is large enough to pass the class information that the backbone provides, further parameters have nothing new to carry. The dataset difference fits this reading, since depth helps EuroSAT about twice as much as CIFAR-10 at eight qubits. One hypothesis is that the eight-angle bottleneck discards more class information from EfficientNet-B0 features on CIFAR-10 than on EuroSAT. We did not run a wider-bottleneck ablation, and Swin-T reaches $86.44\%$ through the same bottleneck, so this explanation remains a hypothesis.

\subsection{What Entanglement Does Here}
\label{sec:disc-entangle}
The product pattern turns the circuit into $n$ independent one-qubit maps, so each output $o_q$ is a smooth trigonometric function of one angle $a_q$ and its own parameters. With frozen features and a linear head, such a per-feature nonlinearity is a natural fit, while a ring mixes features through CNOT gates before the head can weight them. This mechanism is consistent with the accuracy order of Fig.~\ref{fig:entangle}, but our design does not isolate it. A classical model that applies the same per-feature map without a circuit would test it directly. We read the result as a statement about this hybrid design with frozen features, not as a general claim about entanglement in quantum models.

\subsection{A Practical Protocol}
\label{sec:protocol}
We condense the measurements into a protocol that pairs each question about a hybrid QNN with the quantity that answers it, as reported in Table~\ref{tab:protocol}. The central recommendation is a division of labor in which geometry describes how many directions the circuit has and how evenly it uses them. The gradient variance then describes trainability as the register grows, and accuracy or a feature probe describes which backbone to use.

\begin{table}[!h]
\centering
\caption{Practical protocol pairing each question about a hybrid QNN with the quantity that answers it, and the evidence for that pairing in this paper.}
\label{tab:protocol}
\resizebox{\linewidth}{!}{\begin{tabular}{lll}
\toprule
Question & Quantity & Evidence in this paper \\
\midrule
How many parameters move the outputs? & Active count $r_Z$ against $p-n$ & $r_Z/p$ within 0.04 of $1-1/(3L)$ for $L\ge 2$ (Table~\ref{tab:depth}) \\
\midrule
How evenly is the sensitivity shared? & $\dz$, or $\dfull$ when $F$ is computed & Initial $\dz/p$ between 0.16 and 0.25 (Table~\ref{tab:depth}) \\
\midrule
Can a diagonal scalar stand in for $F$? & $\eta$ and Eq.~\eqref{eq:identity} & Overstates $\dfull$ by $1/(1-\eta^{2})$ (Proposition~\ref{prop:identity}) \\
\midrule
Is the initialization trainable as $n$ grows? & Gradient variance against $n$ & Falls by a factor of $34.5$ from $n=4$ to $16$ on CIFAR-10 (Table~\ref{tab:fits}) \\
\midrule
Which backbone should be used? & Accuracy or a feature probe & Initial $\dz$ spans $6.63$ to $9.13$ while accuracy spans $63.30$ points (Table~\ref{tab:scope}) \\
\midrule
Does the scalar forecast accuracy? & Accuracy curve & Sign unstable in 7 of the 14 configurations (Table~\ref{tab:corr}) \\
\midrule
Did entanglement buy accuracy? & Accuracy across patterns & Product exceeds ring by $5.3$ to $8.2$ points (Fig.~\ref{fig:entangle}) \\
\bottomrule
\end{tabular}}
\end{table}

\section{Scope and Future Directions}
\label{sec:future}
In this section, we first state the scope of the measurements, then list the controls that remain open, and finally outline the extensions that would turn the finite-size observations into scaling statements.

The grid is a noise-free statevector simulation in single precision, and the features were cached after one augmentation pass. Width 16 has initialization measurements only, and trained runs stop at $n=12$, so the gradient-variance fit rests on seven widths. Three seeds keep the paired tests valid, yet they limit the precision of every correlation. The study logs the Pauli-$Z$ scalar rather than the full QFIM, so the relation between the two rests on Propositions~\ref{prop:bound} and~\ref{prop:identity} rather than on a measured $\eta$.

Three controls remain open, and all of them are inexpensive to run with the existing pipeline. First, a classical bottleneck that feeds the same $n$ angles through a fixed nonlinearity into the head would test the per-feature mechanism of Section~\ref{sec:disc-entangle}. Second, a linear probe of $B(x)$ would give a direct classical reference for each backbone. Third, a wider bottleneck would test the data-path hypothesis of Section~\ref{sec:disc-data}. We plan to measure the full QFIM on the same grid and to check it against the rank limit of Section~\ref{sec:disc-diag}. The off-diagonal mass $\eta$ and the identity of Proposition~\ref{prop:identity} can then be read from data. We will also repeat the measurements under finite shots and device noise and extend the width sweep with a memory-efficient simulator.

\section{Conclusion}
\label{sec:conclusion}
We showed what the effective dimension logged in hybrid QNNs measures by relating it to the QFIM and testing it on a controlled grid. The logged scalar is the PR of Pauli-$Z$ output sensitivities, each bounded by $n$ times the matching QFIM diagonal entry, and any diagonal PR overstates the full one by $1/(1-\eta^{2})$. Because the last $R_z$ gate on each qubit is invisible to the readout, the active fraction is capped at $1-1/(3L)$. The measured fill sits within 0.04 of that ceiling for two or more layers. At a fixed circuit, CIFAR-10 accuracy spans 23.14\% to 86.44\% across backbones while the scalar stays between 6.63 and 9.13. Its correlation with accuracy during training changes sign across seeds in 7 of the 14 configurations. The gradient variance never meets the plateau rule, yet it falls by a factor of 34.5 from 4 to 16 qubits, and removing entanglement raises accuracy by 5.3 to 8.2 points.

The practical use is a division of labor among diagnostics. Practitioners can use the active count and a PR to inspect the circuit, the gradient variance against width to judge trainability, and accuracy or a feature probe to choose a backbone. A diagonal scalar should be quoted only together with its off-diagonal mass, and a full-matrix implementation should pass the rank check of the product pattern. Future work will measure the full QFIM on the same grid and repeat the study under finite shots and device noise. It will also add a classical bottleneck baseline and extend the width sweep beyond 16 qubits so that the finite-size trends become scaling statements.

\section*{CRediT Authorship Contribution Statement}
M.M.M.: Conceptualization, Methodology, Software, Formal Analysis, Investigation, Data Curation, Visualization, Writing: Original Draft.
A.K.: Supervision, Validation, Project Administration, Writing: Review and Editing.
M.A.A.: Supervision, Resources, Writing: Review and Editing.

\section*{Data Availability}
No new data were collected. CIFAR-10 (\url{https://www.cs.toronto.edu/~kriz/cifar.html}) and EuroSAT (\url{https://github.com/phelber/eurosat}) are public.

\section*{Funding}
This research received no external funding.

\section*{Competing Interests}
The authors declare that they have no competing interests.

\bibliographystyle{IEEEtran}
\bibliography{references}

\clearpage
\appendix
\renewcommand{\thesection}{S\arabic{section}}
\renewcommand{\thetable}{S\arabic{table}}
\renewcommand{\thefigure}{S\arabic{figure}}
\renewcommand{\theHsection}{S\arabic{section}}
\renewcommand{\theHtable}{S\arabic{table}}
\renewcommand{\theHfigure}{S\arabic{figure}}
\setcounter{section}{0}
\setcounter{table}{0}
\setcounter{figure}{0}
\pdfbookmark[1]{Supplementary Material}{supplementary}
\begin{center}
{\Large\bfseries Supplementary Material\par}
\end{center}
\vspace{0.5em}
This supplementary material reports the per-configuration statistics and supporting figures behind Sections~\ref{sec:rq1} to~\ref{sec:rq4}. Section~\ref{app:dyn} covers training dynamics at eight qubits, and Section~\ref{app:backbones} covers the backbone comparison.

\section{Training Dynamics at Eight Qubits}
\label{app:dyn}
In this section, we report the trajectory statistics of the logged scalar on sub-grid A, the training-time gradient variance and active count, and a mixed scalar that appears in practice.

The trajectory statistics of $\dz$ for all 14 configurations of sub-grid A appear in Table~\ref{tab:dyn}. Untrained circuits already carry a moderate $\dz$ well below $p$, and training moves it without driving it toward $p$. One layer contracts the scalar on both datasets, three layers expand it, and four layers end below their initial value. We define the stabilization epoch as the first checkpoint at which the last five values of $\dz$ have a relative std below $0.10$. It lies between 20 and 47 for the learned backbones and reaches 117 for RandProj on CIFAR-10.

\begin{table}[!ht]
\centering
\caption{Trajectory statistics of $\dz$ and the per-run Pearson correlation $\rho$ with test accuracy for the 14 configurations of sub-grid A ($n=8$). Scalar entries are the mean $\pm$ std over three seeds; $\rho$ is the mean with the lowest and highest seed in brackets.}
\label{tab:dyn}
\resizebox{\linewidth}{!}{\begin{tabular}{llccccccc}
\toprule
Dataset & Backbone & $L$ & Initial $\dz$ & Peak $\dz$ & Mean $\dz$ & Late $\dz$ & Stab.\ epoch & $\rho$ [min, max] \\
\midrule
\multirow{7}{*}{CIFAR-10}
 & RandProj & 2 & 8.46\sd{0.20} & 18.72\sd{3.03} & 11.05\sd{2.04} & 11.42\sd{2.69} & 117 & $+0.12$ [$-0.45$, $+0.42$] \\
 & EfficientNet-B0 & 2 & 9.13\sd{0.68} & 10.09\sd{2.19} & 6.97\sd{1.46} & 8.94\sd{2.55} & 43 & $+0.62$ [$+0.38$, $+0.87$] \\
 & ResNet-50 & 2 & 6.63\sd{1.17} & 6.21\sd{0.42} & 4.49\sd{0.97} & 5.34\sd{0.77} & 45 & $+0.40$ [$-0.33$, $+0.82$] \\
 & Swin-T & 2 & 9.08\sd{1.29} & 9.78\sd{2.15} & 6.77\sd{1.22} & 7.19\sd{2.64} & 25 & $+0.13$ [$-0.53$, $+0.54$] \\
 & EfficientNet-B0 & 1 & 5.76\sd{0.30} & 2.17\sd{0.27} & 1.63\sd{0.26} & 1.77\sd{0.50} & 20 & $+0.10$ [$-0.24$, $+0.45$] \\
 & EfficientNet-B0 & 3 & 11.55\sd{0.30} & 16.30\sd{1.67} & 10.64\sd{1.94} & 12.96\sd{0.48} & 32 & $+0.33$ [$+0.14$, $+0.51$] \\
 & EfficientNet-B0 & 4 & 15.23\sd{2.28} & 14.64\sd{0.89} & 9.82\sd{0.47} & 8.13\sd{2.34} & 30 & $+0.20$ [$+0.15$, $+0.27$] \\
\midrule
\multirow{7}{*}{EuroSAT}
 & RandProj & 2 & 8.03\sd{1.08} & 10.80\sd{3.13} & 6.04\sd{0.92} & 6.89\sd{1.55} & 53 & $+0.09$ [$-0.15$, $+0.50$] \\
 & EfficientNet-B0 & 2 & 9.28\sd{0.23} & 9.79\sd{1.07} & 6.14\sd{1.45} & 5.46\sd{3.72} & 42 & $+0.13$ [$-0.41$, $+0.74$] \\
 & ResNet-50 & 2 & 6.18\sd{0.52} & 5.68\sd{1.98} & 3.53\sd{1.05} & 3.92\sd{1.87} & 38 & $+0.27$ [$+0.17$, $+0.45$] \\
 & Swin-T & 2 & 8.90\sd{1.51} & 5.66\sd{1.10} & 4.23\sd{0.65} & 4.79\sd{0.77} & 47 & $+0.45$ [$+0.30$, $+0.76$] \\
 & EfficientNet-B0 & 1 & 6.03\sd{0.71} & 3.71\sd{0.29} & 2.79\sd{0.11} & 2.56\sd{0.17} & 42 & $-0.27$ [$-0.72$, $+0.20$] \\
 & EfficientNet-B0 & 3 & 12.65\sd{1.02} & 14.81\sd{1.58} & 11.21\sd{1.54} & 11.54\sd{2.41} & 32 & $+0.61$ [$+0.50$, $+0.70$] \\
 & EfficientNet-B0 & 4 & 15.30\sd{1.63} & 11.46\sd{1.13} & 8.85\sd{1.17} & 8.52\sd{1.32} & 32 & $+0.50$ [$+0.31$, $+0.62$] \\
\bottomrule
\end{tabular}}
\end{table}

The untrained, peak, and late values of $\dz$ at $L=2$ appear in Fig.~\ref{fig:geomband}, and the training-time gradient variance and active count appear in Fig.~\ref{fig:trainvar}. A mixed scalar that multiplies the trace-to-norm ratio of $\fz$ by $\sqrt{r_Z}$ also appears in practice, and it equals $\sqrt{\dz\,r_Z}$, which lies between $\dz$ and $r_Z$. It preserves the ordering of checkpoints closely, as shown in Fig.~\ref{fig:mixed}, but it blends two quantities and cannot be read as a count of directions.

\begin{figure}[!ht]
\centering
\includegraphics[width=0.62\textwidth]{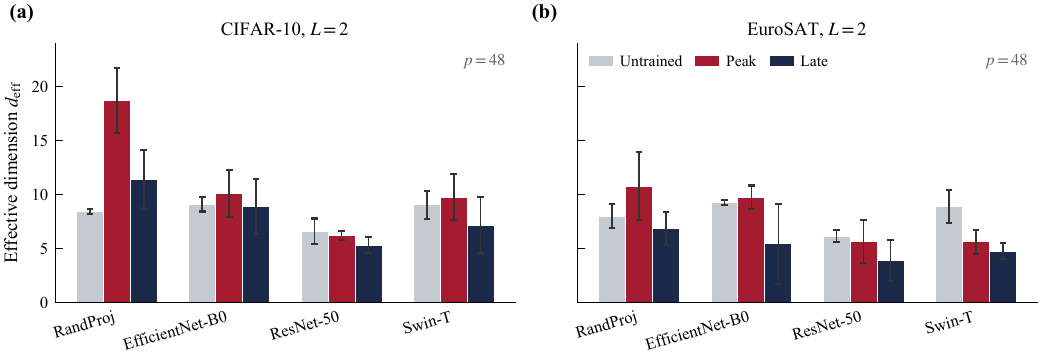}
\caption{Untrained, peak, and late $\dz$, labeled effective dimension on the axis, at $n=8$ and $L=2$ for four backbones on (a)~CIFAR-10 and (b)~EuroSAT, as the mean and std over three seeds.}
\label{fig:geomband}
\end{figure}

\begin{figure}[!ht]
\centering
\includegraphics[width=0.72\textwidth]{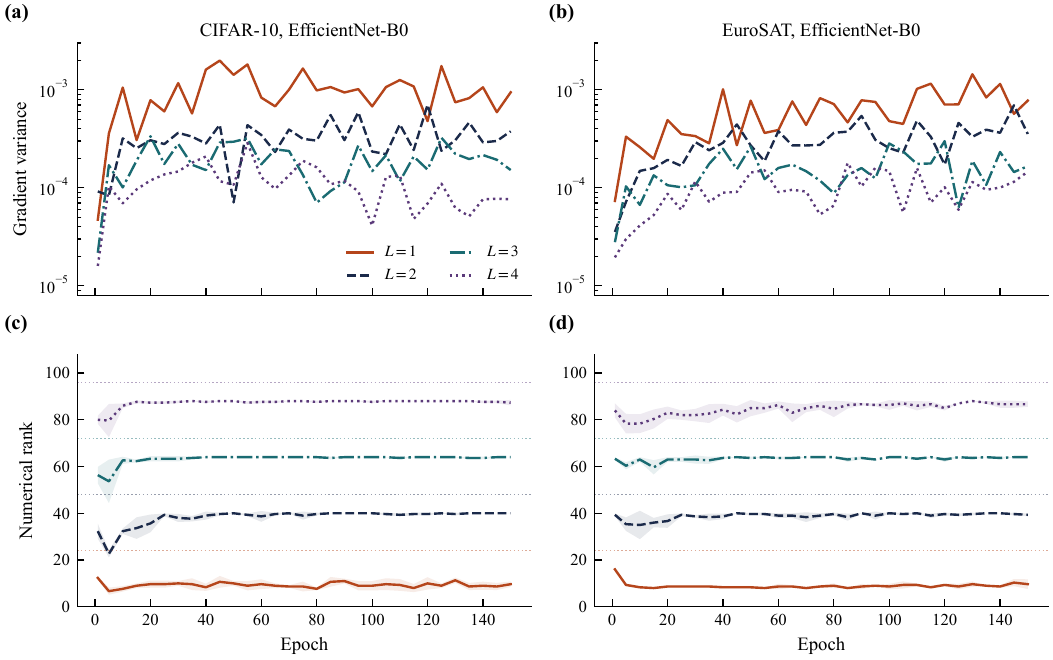}
\caption{Training-time behavior at $n=8$ (EfficientNet-B0, sub-grid A), as seed means for $L\in\{1,2,3,4\}$. (a,b)~Gradient variance against epoch on a logarithmic axis. (c,d)~Active count $r_Z$, labeled numerical rank on the axis, against epoch; dashed lines mark $p=24L$.}
\label{fig:trainvar}
\end{figure}

\begin{figure}[!ht]
\centering
\includegraphics[width=0.5\textwidth]{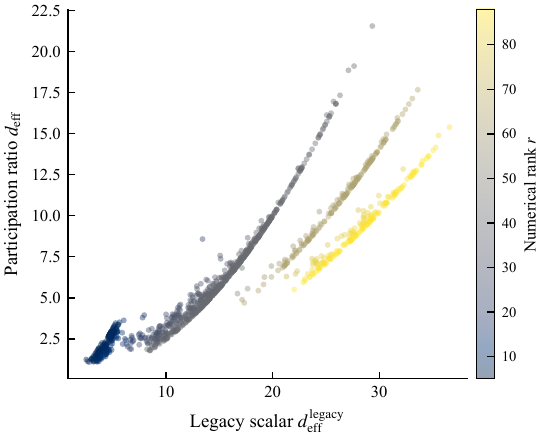}
\caption{Participation ratio $\dz$ against the mixed scalar $\sqrt{\dz\,r_Z}$, labeled legacy scalar on the axis, over all logged checkpoints of sub-grid A. Color is the active count $r_Z$.}
\label{fig:mixed}
\end{figure}

\section{Backbone Comparison}
\label{app:backbones}
In this section, we report the paired tests behind the backbone order of Section~\ref{sec:rq4} and the matching bars of geometry and accuracy.

The paired accuracy gaps between backbones appear in Table~\ref{tab:sig}. Every paired $t$-test is significant at the five percent level, while the Wilcoxon test sits at its floor of $0.25$ in every comparison and serves only as a sign check. Paired bars of the initial $\dz$ and final accuracy for the four backbones at $n=8$ appear in Fig.~\ref{fig:dissoc-dz}.

\begin{table}[!ht]
\centering
\caption{Paired accuracy gaps between backbones at $n=8$ and $L=2$ (sub-grid A), in percentage points of final test accuracy over three paired seeds. An asterisk marks $P<0.05$ under the paired $t$-test.}
\label{tab:sig}
\begin{tabular}{llcccc}
\toprule
Dataset & Comparison & Gap (points) & $t$ & $P$ ($t$-test) & $P$ (Wilcoxon) \\
\midrule
\multirow{6}{*}{CIFAR-10}
 & EfficientNet-B0 $-$ RandProj & $+51.11^{*}$ & 58.47 & 0.0003 & 0.250 \\
 & ResNet-50 $-$ RandProj & $+59.09^{*}$ & 77.00 & 0.0002 & 0.250 \\
 & Swin-T $-$ RandProj & $+63.30^{*}$ & 74.95 & 0.0002 & 0.250 \\
 & ResNet-50 $-$ EfficientNet-B0 & $+7.97^{*}$ & 17.32 & 0.0033 & 0.250 \\
 & Swin-T $-$ EfficientNet-B0 & $+12.19^{*}$ & 13.92 & 0.0051 & 0.250 \\
 & Swin-T $-$ ResNet-50 & $+4.22^{*}$ & 10.11 & 0.0096 & 0.250 \\
\midrule
\multirow{6}{*}{EuroSAT}
 & EfficientNet-B0 $-$ RandProj & $+43.02^{*}$ & 16.12 & 0.0038 & 0.250 \\
 & ResNet-50 $-$ RandProj & $+53.19^{*}$ & 15.21 & 0.0043 & 0.250 \\
 & Swin-T $-$ RandProj & $+56.00^{*}$ & 19.20 & 0.0027 & 0.250 \\
 & ResNet-50 $-$ EfficientNet-B0 & $+10.17^{*}$ & 10.12 & 0.0096 & 0.250 \\
 & Swin-T $-$ EfficientNet-B0 & $+12.98^{*}$ & 23.37 & 0.0018 & 0.250 \\
 & Swin-T $-$ ResNet-50 & $+2.81^{*}$ & 4.83 & 0.0403 & 0.250 \\
\bottomrule
\end{tabular}
\end{table}

\begin{figure}[!ht]
\centering
\includegraphics[width=0.72\textwidth]{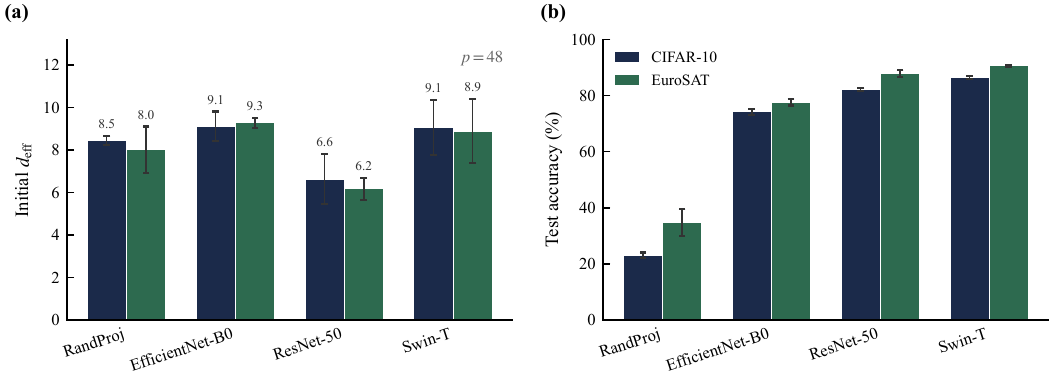}
\caption{Four backbones at $n=8$ and $L=2$ (sub-grid A). (a)~Initial $\dz$, labeled $d_{\mathrm{eff}}$ on the axis. (b)~Final test accuracy. Bars are seed means and error bars are one std over three seeds.}
\label{fig:dissoc-dz}
\end{figure}

\end{document}